\documentclass{siamltex}

\usepackage{amsfonts, amsmath, amssymb, psfrag, graphicx, bbm, mathrsfs}

\usepackage{soul}
\usepackage{ dsfont }
\usepackage[colorlinks=true]{hyperref}
\hypersetup{urlcolor=blue, citecolor=blue, linkcolor=blue}
\usepackage[left=2.5cm,right=2.5cm]{geometry}

\newtheorem{assumption}{Assumption}

\newtheorem{theoremold}{Theorem}

\def\fin{\ifmmode{\Large$\diamond$}\else{\unskip\nobreak\hfil
    \penalty50\hskip1em\null\nobreak\hfil{\Large$\diamond$}
    \parfillskip=0pt\finalhyphendemerits=0\endgraf}\fi}

\def\be#1#2\ee{\begin{equation}\label{eq:#1}#2\end{equation}}
\def\req#1{{\rm(\ref{eq:#1})}}
\def\bdm  {\begin{displaymath}}
  \def\edm  {\end{displaymath}}
\def\bdmal{\begin{displaymath}\begin{aligned}}
    \def\edmal{\end{aligned}\end{displaymath}}

\mathchardef\PhiG="0108

\newcommand{\N}{{\mathord{\mathbb N}}}
\newcommand{\dist}{\textnormal{dist} }
\newcommand{\R}{{\mathord{\mathbb R}}}

\renewcommand{\P}{{\mathord{\mathbb P}}}
\newcommand*\diff{\mathop{}\!\mathrm{d}}

\newcommand{\qb}{\Bar{q}}

\newcommand{\xx}{\textit{\textbf{x}}}
\newcommand{\yy}{\textit{\textbf{y}}}

\newcommand{\rmd}{\,\mathrm{d}}

\newcommand{\dr}{\rmd r}

\newcommand{\dx}{\rmd x}

\newcommand{\dxx}{\rmd \xx}
\newcommand{\dyy}{\rmd \yy}

\def\req#1{{\rm(\ref{eq:#1})}}

\makeatletter
\newcommand{\dupdots}{\mathinner{\mkern1mu\raise\p@
    \vbox{\kern7\p@\hbox{.}}\mkern2mu
    \raise4\p@\hbox{.}\mkern2mu\raise7\p@\hbox{.}\mkern1mu}}
\makeatother
\makeatletter
\newcount\c@MaxMatrixCols \c@MaxMatrixCols=10

\makeatother

 {\endMakeFramed}

\newcommand{\dP}{\,\rmd\P}

\makeatletter
\renewcommand\@biblabel[1]{#1.}
\makeatother
\title{An inverse cluster expansion for the chemical potential
\thanks{The research leading to this work
    has been done within the 
    Collaborative Research Center TRR~146; corresponding funding 
    by the DFG is gratefully acknowledged.}}
\author{Fabio Frommer\thanks{Institut f\"ur Mathematik, Johannes
    Gutenberg-Universit\"at Mainz, 55099 Mainz, Germany
    ({\tt fabiofrommer@uni-mainz.de})}}

\begin{document}
\sloppy
\maketitle
\begin{abstract}
Interacting particle systems in a finite-volume in equilibrium are often described by a grand-canonical ensemble induced by the corresponding Hamiltonian, i.e.\ a finite-volume Gibbs measure. However, in practice, directly measuring this Hamiltonian is not possible, as such, methods need to be developed to calculate the Hamiltonian potentials from measurable data. In this work, we give an expansion of the chemical potential in terms of the correlation functions of such a system in the thermodynamic limit. This is a justification of a formal approach of Nettleton and Green from the 50's, that can be seen as an inverse cluster expansion.
\end{abstract}


\section{Introduction}
In classical equilibrium statistical physics, interacting particle systems are often described via so-called Gibbs measures, either in a finite or infinite-volume. A possible arrangement of particles (also called configuration) has a higher probability when it is energetically favorable. This energy gets calculated via a Hamiltonian that depends on the chemical potential and the interaction potentials associated to the particles. However, when working with actual data, these potentials are impossible to measure. What instead can be measured or extrapolated from simulation data are the so-called correlation functions of the underlying Gibbs measure. This gives rise to the inverse problem: Given the correlation functions, determine the chemical potential and the interaction potentials.
\\
For the chemical potential that induces a given density this problem was first discussed by Chayes, Chayes and Lieb in \cite{Chayes83} and \cite{Chayes84}. The case of higher-order interactions is also briefly discussed. \\
For the pair-interaction case the inverse problem was investigated in \cite{Henderson74} in the case of the finite-volume and in  \cite{Koral07} and \cite{Frommer19} these results were extended upon and made rigorous in the thermodynamic limit, i.e.\ in infinite-volume.
However, all these results do not give a constructive way of recovering the interaction. Estimation methods for the pair-interaction were developed by Takacs and Fiksel in \cite{Takacs86} and \cite{Fiksel84} and have since been expanded upon for more general Hamiltonians in e.g.\ \cite{Coeurjolly11}, another approach based on variational methods was developed in \cite{Baddeley13}, however, it cannot be used to estimate the chemical potential. A first result of an inverse expansion of the chemical potential is found in \cite{Lebowitz64} where the expansion is given in terms of the one-point correlation function (the density) and the Mayer function, see also \cite{jansen23} for an extension of this result to the non-homogeneous case. For this formula however prior knowledge of the pair-interaction is required. The goal of this paper is to find an expansion of the chemical potential, that solely depends on functions derived from the correlation functions of the Gibbs measure. \\
The idea we use goes back to \cite{Nettleton57}
, where an expansion of this type is used in finite-volume. However, the ansatz used in these works constructs multi-body interactions and thus showing convergence of these inversion formulas cannot be done using the classical tools of cluster expansions. \\
In this work, we will use a different approach to prove convergence for the explicit inversion formula for the chemical potential by using an exponential representation formula. 
We then use Bell-polynomials to find bounds on the appearing coefficients, which have been used for a similar scope independently in \cite{Nguyen20} to get bounds on the virial coefficients. In fact, we can show that this formula converges in both finite and infinite-volume. In particular, as we only make very mild assumptions on the Hamiltonian, which include many classic models, and some integrability conditions on the correlation functions this expansion holds not only for pair-interactions, but also in some cases of multi-body interactions of arbitrary order.\\
The outline is as follows: In Section \ref{sec:setting} we introduce the general setting of Gibbs measures we work with and introduce the main idea of the expansion. Then, we formulate in Section \ref{sec:maintresult} our main assumptions as well as the main result and give a few examples in which these assumptions hold in Section \ref{sec:examples}. The last two Sections are devoted to the proof of the main result, in Section \ref{sec:firststep}  we start by proving a result about the zero-point density, which then gets generalized in Section \ref{sec:proof} to the one-point density to obtain the expansion for the chemical potential.

\section{Setting}\label{sec:setting}
A point process $\P$ is a probability measure on the set of configurations 
\begin{align*}
    \Gamma = \{\eta \subset \R^d \mid N(\eta \cap \Delta )< \infty \text{ for all } \Delta \subset \R^d \text{ bounded}\}
\end{align*}
equipped with the $\sigma$-algebra $\mathcal{F}:=\sigma\left(\{N(\cdot\cap \Delta)=m\}\mid m\in\N_0, \,\,\Delta\subset \R^d \text{ bounded}\right)$. Here $N(\eta)=\#\eta$ is the number of elements of $\eta\subset \R^d$. We denote by $\Gamma_0= \{\gamma \in \Gamma \mid N(\gamma)< \infty\}$ the set of finite configurations. Such a point process is a Gibbs measure for some chemical potential $\mu\in\R$ and a Hamiltonian $H\colon \Gamma\to \R\cup\{\infty\}$, we write $\P$ is a $(\mu,H)$-Gibbs measure, if it is supported on a set of \emph{tempered} configurations, see e.g.\ \cite{Vasseur20} for a discussion, and satisfies the \emph{Ruelle-equation}, see \cite{Ruelle70} (herein called \emph{system of equilibrium equations}), namely if, for every non-negative function $G\colon\Gamma \to [0,\infty]$ and every bounded set $\Lambda\subset \R^d$, there holds
\begin{align*}
    \int_{\Gamma} G(\eta)\dP(\eta) = 
    \int_{\Gamma_{\Lambda^c}}\sum_{n=0}^\infty\frac{e^{n\mu}}{n!}\int_{\Lambda^n} 
    G(\{\xx_n\}\cup\eta){e^{-H(\{\xx_n\})-W(\{\xx_n\}\mid \eta)}}\dxx_n\dP(\eta)
\end{align*}
where $\Lambda^c=\R^d\backslash\Lambda$, $\Gamma_{\Lambda^c} = \{\eta\in\Gamma \mid \eta \subset \Lambda^c\}$ and $W(\{\xx_n\}\mid \eta)$ is the interaction between the part of the particles inside $\Lambda$ and those on the outside. If $H$ has finite-range $W$ can be defined as
\begin{align*}
    W(\{\xx_n\}\mid \eta):=
    H(\{\xx_n\}\cup \eta)
    -H(\{\xx_n\})-H(\eta), 
    \qquad \xx_n\in \Lambda^n, \,\eta\in\Gamma_{\Lambda^c}
\end{align*}
with the convention that we set $W(\{\xx_n\}\mid \eta)=+\infty $ if $H(\{\xx_n\}\cup \eta)=+\infty$. For general Hamiltonians it is possible to extend this definition, cf.\ \cite{Vasseur20}. However, the specific definition of $W$ is not important for this work.
We also assume that the Hamiltonian $H$ is \emph{translation-invariant}, i.e.\ for every $x\in \R^d$ we have that $H(\eta)= H(\tau_x(\eta))$ where $\tau_x(\eta) := \{y-x \mid y\in\eta\} $, that $H$ is \emph{hereditary}, i.e.\ that $H(\eta \cup \{x\})= +\infty $ for all $x\in \R^d$ whenever $H(\eta) = +\infty$ and that $H$ does not contain any contribution of a self-interaction of particles, i.e.\ $H(\{x\})=0$ for all $x\in\R^d$. The last assumption we make on $H$ is, that $H$ is \emph{stable},
i.e.\ there is a constant $B>0$ such that 
\begin{align*}
    H(\{\xx_n\}) \geq -Bn.
\end{align*}
In this case, under some additional technical assumptions, e.g. {finite-range} of the interaction or the Hamiltonian consists of a {regular pair-interaction}, there is at least one translation-invariant $(\mu,H)$-Gibbs measure $\P$ associated to these parameters, cf.\ \cite{Vasseur20}. In the sequel, we will assume that $\P$ is such a translation-invariant Gibbs measure.
A family of functions $(j_\Lambda^{(n)})_{n  \geq 0, \,\Lambda \subset \R^d \text{ bounded}}$ are called the \emph{finite-volume densities} of $\P$, or \emph{Janossy densities}, if for every  bounded set $\Lambda\subset \R^d$ and every non-negative function $G\colon \Gamma \to [0,\infty]$ there holds 
\begin{align*}
	\int_{\Gamma} G(\eta\cap \Lambda)\dP(\eta) = 
	\sum_{n=0}^\infty \frac{1}{n!}
	\int_{\Lambda^n}
	G(\{\xx_n\})j_\Lambda^{(n)}(\xx_n)\dyy_k.
\end{align*}
The Ruelle-equation characterizes the Janossy densities, they exist for every $(\mu,H)$-Gibbs measure $\P$ and for fixed $\Lambda\subset \R^d$ they are given by
\begin{align}\label{eq:JanossyHamiltonian}
    j_\Lambda^{(n)}(\xx_n) =  \int_{\Gamma_{\Lambda^c}} \exp\Big(n\mu-H(\{\xx_n\})-W(\{\xx_n\}\mid \eta)\Big) \dP(\eta), \qquad n \in \mathbb{N}_0.
\end{align}
Note that by definition the Janossy densities are \emph{symmetric} functions, i.e.\
for every permutation $\sigma \in \mathcal{S}_n$ there holds $ j_\Lambda^{(n)}(x_1,\dots,x_n)=  j_\Lambda^{(n)}(x_{\sigma(1)},\dots,x_{\sigma(n)})  $.
Another family of symmetric functions we are going to need are the \emph{correlation functions}
$(\rho^{(n)})_{n\geq 0} $, where $\rho^{(0)}:=1$; when a point process has Janossy densities they are known to exist and for $\xx_n\in (\R^d)^n$ the $n$-point correlation function is given by
\begin{align*}
    \rho^{(n)}(\xx_n) = \sum_{k=0}^\infty \frac{1}{k!}\int_{\Lambda^n}j_\Lambda^{(n+k)}(\xx_n,\yy_k)\dyy_k,
\end{align*}
for any $\Lambda$ such that $\xx_n \in\Lambda^n$. Note that we write $j_\Lambda^{(n+k)}(\xx_n,\yy_k) = j_\Lambda^{(n+k)}(x_1,\dots,x_n,y_1,\dots,y_k)$ for simplicity and that the Janossy densities contain the averaged contributions of the whole space by \req{JanossyHamiltonian} and thus the right-hand side above is independent of $\Lambda$, if $\xx_n \in\Lambda^n.$
In particular, if $\P$ is translation-invariant, then $\rho^{(1)}=\rho$ is constant.
Under the additional assumption that $\P$ satisfies a so-called \textit{Ruelle-condition}, i.e.\ there is a $\xi>0 $ such that 
\begin{align*}
    0 \leq \rho^{(n)}(\xx_n) \leq \xi^n  \quad\text{ for all } \xx_n \in (\R^{d})^n, \quad n\geq 1
\end{align*}
there also holds the well-known inverse formula
\begin{align}\label{eq:Janossycorr}
     j_\Lambda^{(n)}(\xx_n) = \sum_{k=0}^\infty
     \frac{(-1)^k}{k!}\int_{\Lambda^k}\rho^{(n+k)}(\xx_n,\yy_k)\dyy_k, \qquad \xx_n \in \Lambda^n,
\end{align}
see e.g.\ \cite{Ruelle70}. 
\begin{remark}
In experiments, when only having samples of a Gibbs measure, it is straightforward to estimate the correlation functions of the process without having to make prior assumptions on the Hamiltonian. Thus, the correlation functions are the experimentally easier to obtain starting object in contrast to the Janossy densities.
\end{remark}\\
The idea of this work is to find an explicit expansion of $\log  j_\Lambda^{(1)}$, using the identity \req{Janossycorr}, in terms that only depend on the correlation functions $(\rho^{(n)})_{n\geq 0}$. From this, we can then obtain an explicit formula for the chemical potential $\mu$. The reason is as follows: \\
By \req{JanossyHamiltonian} there holds
\begin{align*}
    \log j_\Lambda^{(1)}(x)= \mu + \log \int_{\Gamma_{\Lambda^c}} \exp\Big(-W(\{x\}\mid \eta)\Big) \dP(\eta). 
\end{align*}
Now it is well-known that both sides of the above equation diverge to $-\infty$ for $\Lambda\nearrow \R^d$. 
In the sequel we will make the following technical assumption, that can be interpreted as an implicit decay condition on the Hamiltonian and in particular is obvious for finite-range potentials,
\begin{align*}
    \log \int_{\Gamma_{\Lambda^c}} \exp\Big(-W(\{x\}\mid \eta)\Big) \dP(\eta)
    =
    \log \int_{\Gamma_{\Lambda^c}}  \dP(\eta)+
    \varepsilon(\Lambda),
\end{align*}
where $\varepsilon(\Lambda)\to 0 $ as $\Lambda \nearrow\R^d$. The first part of the right-hand side is equal to $\log j_\Lambda^{(0)}$ and thus $\log j_\Lambda^{(1)}(x)-\log j_\Lambda^{(0)} \to \mu$ as $\Lambda\nearrow\R^d$.

Our main tool is a result about exponential representation that has been used by Ruelle in \cite[Section 4.4]{Ruelle69} for the correlation functions, see also \cite{Reb14} for a more general version, which we will state here for ease of reading.
\begin{theoremold}[Exponential representation \cite{Reb14}]\label{thm:reb}
Let $(F_n)_{n\geq 0}$ be a family of symmetric functions with $F_n\colon (\R^d)^n\to \R$, $F_0\equiv 1$ and there exist $0<c<1/2$ and $D>0$ such that for every bounded $\Lambda\subset\R^d$ there holds
\begin{align}\label{eq:expbound}
    \int_{\Lambda^n}|F_n(\xx_n)|\dxx_n \leq |\Lambda| n! D  c^n 
\end{align}
for every $n\in\N$, here $|\Lambda|$ denotes the Lebesgue measure of the set $\Lambda$. 
Then the function $\Phi:\Gamma_0\to \R$ defined by
\begin{align}\label{eq:partitionrep}
    \Phi(\eta) = \sum_{k=1}^{N(\eta)} \sum_{\pi \in\Pi_k(\eta)}\prod_{i=1}^k F_{\kappa_i}({\pi_i}),
\end{align}
where $\Pi_k(\eta) $ denotes the set of partitions of $\eta$ into $k$ non-empty sets $\pi_i$ $i=1,\dots,k$ where $\pi_i$ consists of $\kappa_i$ elements,
satisfies 
\begin{align}\label{eq:exprep}
    \sum_{n=0}^\infty \frac{1}{n!} \int_{\Lambda^n} \Phi(\xx_n)\dxx_n = \exp\left(\sum_{n=1}^\infty \frac{1}{n!} \int_{\Lambda^n}F_n(\xx_n) \dxx_n\right).
\end{align}
\end{theoremold} 
The last tool we will need are the so-called \emph{truncated correlation functions}.
\begin{definition}
For a point process with correlation functions $(\rho^{(n)})_{n\geq 0}$, the \emph{truncated correlation functions} $(\rho^{(n)}_T)_{n\geq 1}$ are defined recursively. For $n=1$ we define $\rho_T^{(1)}= \rho^{(1)}=\rho$ and for $n\geq 2$ 
\begin{align}\label{eq:clusterfunctions}
    \rho^{(n)}_T(\xx_n)= \rho^{(n)}(\xx_n)-\sum_{k=2}^n \sum_{\pi \in\Pi_k(\{\xx_n\})}\prod_{i=1}^k \rho_T^{(\kappa_i)}({\pi_i}). 
\end{align}
\end{definition}

\section{Main result}\label{sec:maintresult}
We will now state our main result, which will be proved in Section \ref{sec:proof}. 
To formulate our result we introduce the following two assumptions:
\begin{assumption}\label{ass:A}
Let $W$ be the interaction associated to the Hamiltonian $H$. We assume for any bounded set $\Delta\subset \R^d$ and every $(\mu,H)$-Gibbs measure $\P$ there holds
\begin{align*}
    \lim_{\Lambda\nearrow \R^d}\sup_{x \in \Delta}\left|\log \int_{\Gamma_{\Lambda^c}} \exp\Big(-W(\{x\}\mid \eta)\Big) \dP(\eta)
    -
    \log \int_{\Gamma_{\Lambda^c}}  \dP(\eta)\right| = 0.
\end{align*}
\end{assumption}
\begin{assumption}\label{ass:B}
There are constants $D>0$ and $q>0$ such that for the truncated correlation functions of $\P$ there holds
\begin{align}\label{eq:Rboundcluster}
    \sup_{x\in\R^d}\int_{\Lambda^n}  \frac{1}{\rho}\left|{\rho}_{T}^{(1+n)}(x,\yy_n) \right|\dyy_n
    \leq n! D q^{n}, 
\end{align}
for every bounded set $\Lambda\subset \R^d$ and every $n\in\N$.     
\end{assumption}
\begin{theorem}\label{thm:mainthm}
Let $\P$ be a $(\mu,H)$-Gibbs measure that satisfies a Ruelle-conditions and Assumptions \ref{ass:A} and \ref{ass:B}.
If $q < q_0$, where 
\begin{align}\label{eq:qzero}
    q_0 =  \frac{1}{2(2+\zeta D)}
    \qquad \text{ with }\qquad \zeta= \frac{1}{2\log 2-1},
\end{align}
then there holds 
\begin{align}\label{eq:muexpansionfinal}
    \mu = \log \rho  +  
   \sum_{k=1}^\infty\frac{(-1)^k}{k!} \int_{(\R^d)^k}
    \widetilde{\rho}^{\,(1+k)}_T(0,\yy_k) \dyy_k
\end{align}
where the family $(\widetilde{\rho}^{\,(1+k)}_T)_{k\geq 1}$ is recursively defined by $\widetilde{\rho}^{\,(2)}_T(x,y) = {\rho}^{\,(2}_T(x,y)/\rho$ and for $k\geq 2$ by 
\begin{align}\label{eq:rhotilde}
    \widetilde{\rho}^{(1+k)}_T(x,\yy_k) =
    \frac{\rho^{(1+k)}_T(x,\yy_k)}{\rho}
    -\sum_{l=2}^k \sum_{\pi\in\Pi_l(\{\yy_k\})} \prod_{i=1}^l\widetilde{\rho}^{\,(1+\kappa_i)}_T (x,\pi_i).
 \end{align}
\end{theorem}
\begin{remark}
The functions $(\widetilde{\rho}^{\,(1+k)}_T)_{k\geq 1}$ take a similar role as the so-called \emph{Ursell functions} in classical cluster expansions and thus the expansion above can be seen as a type of multi-body cluster expansion.
\end{remark}

\section{Examples}\label{sec:examples}
We will now take a look at three cases in which the assumptions of Theorem \ref{thm:mainthm} are satisfied. First, we will look at a classic superstable pair-interaction, then at a  Hamiltonian consisting of multi-body potentials of arbitrary order, which are non-negative and of finite-range, in both settings there always exists at least one translation-invariant $(\mu,H)$-Gibbs measure, cf.\ \cite{Vasseur20}. Finally, we will look at the Kirkwood-closure process from \cite{Kuna07} which can also be shown to be a Gibbs point process for a Hamiltonian consisting of multi-body potentials of every order.
\subsection{Superstable pair-interaction}\label{ex:superstab}
Let us consider the case that $H$ consists only of a pair-interaction, i.e.\
a measurable even function $u\colon\R^d\to\R\cup\{+\infty\}$ such that
$H(\{\xx_n\}) = \sum_{i<j}u(x_i-x_j)$.
If we further assume that $u$ is \emph{superstable}, i.e.\ that there exists $r_0>0$ and 
decreasing positive functions $\varphi:(0,r_0)\to \R^+_0$ and 
$\psi: [0,\infty)\to \R^+$ with
\begin{align*}
     \int_{0}^{r_0} r^{d-1}\varphi(r)\dr = +\infty
   \qquad \text{and} \qquad
   \int_{0}^\infty r^{d-1}\psi(r)\dr <\infty\,,
\end{align*}
and
\begin{align*}
   u(x) &\,\geq\, \varphi(|x|) \qquad \textrm{ for }\ 0 < |x| < r_0\,, \\[1ex]
   |u(x)| &\,\leq\, \psi(|x|) \qquad \textrm{ for }\ |x| \geq r_0\,.
\end{align*}
It is well-known, cf.\ Chapter 4 of \cite{Ruelle69} that $\P$ satisifies Assumption \ref{ass:B} if
\begin{align}\label{eq:classiczbound}
    e^\mu = \frac{\qb}{e^{2B+1}C(u)}
\end{align}
for some $\qb<1/2$, where $C(u):=\int_{\R^d}|e^{-u(x)}-1|\dx$ and $B=B(u)>0$ is the stability constant of $H$, with
\begin{align*}
    D=\frac{e^\mu e^{-2B}}{(1-\qb)}\frac{1-\qb}{1-2\qb}\qquad \text{ and } \qquad q= \frac{\qb}{1-\qb}.
\end{align*}
Let $\P$ be such a $(\mu,u)$-Gibbs measure.
From Theorem 5.5 in \cite{Ruelle70}, it follows immediately that $\P$ satisfies Assumption \ref{ass:A}. Therein, it is also shown that $\P$ satisfies a Ruelle-condition.
We thus get the easy result:
\begin{corollary}
Let  $\P$ be a $(\mu,u)$-Gibbs measure with a superstable pair-interaction $u$, if
\begin{align*}
	e^\mu < \min \left\lbrace
	\frac{1}{3},
	\frac{1}{1+ \frac{2}{2+\zeta D}}
	\right\rbrace\left(e^{2B+1}C(u)\right)^{-1},
\end{align*}    
then the expansion \req{muexpansionfinal} holds for the true chemical potential of $\P$.
\end{corollary}\\
%
%
%

\subsection{Non-negative multi-body interaction of finite-range}
We say a Hamiltonian has \emph{finite-range}, if there exists an $R>0$ such that, whenever $|x_0-y|>R$ there holds
\begin{equation}\label{eq:multibodyfinrange}
	W(\{x_0\} \mid \{\xx_n\}\cup \{y\}) 
	=W(\{x_0\} \mid  \{\xx_n\})
\end{equation}
for every set of points $\{\xx_n\} \subset  \R^d$. In particular, this means that $W(\{x_0\} \mid \eta)=0$, for any $\eta \in \Gamma$ with $\dist(\{x_0\}, \eta)>R$. As mentioned, in this case there always exists at least one $(\mu,H)$-Gibbs measure, see e.g.\ \cite{Vasseur20}. Furthermore, $\P$ satisfies a Ruelle-condition with $\xi=e^\mu$. 
\begin{corollary}
Let $\P$ be a $(\mu,H)$-Gibbs measure where $H$ satisfies \req{multibodyfinrange} as above, then $\P$ satisfies Assumption \ref{ass:A}.
\end{corollary}
\begin{proof}
Fix some $\Delta\subset \R^d$ compact and let $\Lambda \subset \R^d$
be such that $\dist(\Delta, \partial\Lambda) >R$, then \req{multibodyfinrange} and \req{JanossyHamiltonian} for $n=1$ imply
\begin{equation*}
	j_\Lambda^{(1)}(x) =  e^\mu\int_{\Gamma_{\Lambda^c}} \dP(\eta) 
	.
\end{equation*}
This proves the claim.
\end{proof} \\
It was shown by Moraal in \cite{Moraal76} that if $H$ is additionally non-negative, then the multi-body version of the Kirkwood-Salsburg operator is bounded and there exists a solution to the multi-body version of the Kirkwood-Salsburg equation. Using the same techniques as in the pair-interaction case, see Section 4.4.6 of \cite{Ruelle69}, it then follows that $\P$ satisfies Assumption \ref{ass:B}.
\begin{corollary}
Let $H$ be a non-negative Hamiltonian of finite range, $\mu\in\R$ and $\P$ be a $(\mu,H)$-Gibbs measure. If $\mu$ is sufficiently small, then the expansion \req{muexpansionfinal} holds.
\end{corollary}
\begin{remark}
In \cite{Skrypnik08} Skrypnik extended the results of Moraal to the case that $H$ consists of the sum of a very particular non-negative multi-body interaction of finite-range and a superstable pair-interaction. In this case, it can also be shown that Assumption \ref{ass:A} is fulfilled. 
\end{remark}

\subsection{The Kirkwood-closure process}
A point process $\P $ is called the Kirkwood-closure process for a given density $\rho>0$ and a non-negative even function $g$ on $\R^d$ if the correlation functions of $\P$ satisfy
\begin{align}\label{eq:easycorrelations}
    \rho^{(n)}(\xx_n) = \rho^n \prod_{1\leq i<j \leq n} g(x_i-x_j).
\end{align}
It was shown in \cite{Kuna07} by Kuna et al. that, if $C(g):=\int_{\R^d}|g(x)-1|\dx<\infty$ and there is a $b\geq 1$ such that $\prod_{i=1}^n g(x_i-x_0)\leq b$ for any $x_0,x_1,\dots,x_n \in\R^d$ with $\prod_{i<j}g(x_i-x_j)>0$, then the Kirkwood-closure process exists for $\rho< (ebC(g))^{-1}$. This was an extension of the results in \cite{Ambartzumian91}, where only the case $g\leq 1$ was discussed. The Kirkwood-closure, which in computational physics is often called Kirkwood-superposition was first suggested in \cite{Kirkwood35} and has been widely used to approximate higher-order correlation functions and thus just having to calculate $\rho$ and $\rho^{(2)}$ from simulation data.
In this case, it also known that $\P$ satisfies a Ruelle condition with $\xi = \rho b^{1/2}$ and
for every bounded $\Lambda\subset \R^d$
\begin{align*}
	\sup_{x\in\R^d}\int_{\Lambda^n}
	\left|{\rho}_{T}^{(1+n)}(x,\yy_n) \right|\dyy_n
    \leq n! D\left(\rho e b C(g)\right)^{n+1}, 
\end{align*}
i.e.\ 
the Kirkwood-closure process satisfies Assumption \ref{ass:B}, cf.\ Section 3.1 of \cite{Kuna07}. 
In \cite{Gloetzl80} Glötzl gave sufficient conditions to guarantee the existence of a chemical potential $\mu$ and a Hamiltonian $H$ for which the point process is Gibbs. These conditions are for the \emph{Campbell-measure} $\mathrm{C}_{\P}$ of a point process $\P$, i.e.\ the measure on $(\R^d\times\Gamma,\mathcal{B}(\R^d)\otimes \mathcal{F})$ (here $\mathcal{B}(\R^d)$ are the Borel-sets of $\R^d$) defined by 
\begin{align*}
	\mathrm{C}_{\P}(B\times F) := \int_{\Gamma } \sum_{x \in \eta }\mathds{1}_{F}(\eta\backslash\{x\}) \mathds{1}_B(x)\dP(\eta).
\end{align*}
Namely, if the Campbell-measure $\mathrm{C}_{\P}$ is absolutely continuous with respect to $\dx \times \dP$ and the density $k$ satisfies some additional assumptions, such a Hamiltonian exists, see Satz 2 in \cite{Gloetzl80}.
Due to the easy structure \req{easycorrelations} of the correlation functions, a straightforward computation, using \req{Janossycorr}, shows that $\P$ satisfies these assumptions and the Radon-Nikodym derivative of the Campbell-measure is given by
\begin{align*}
    k(x,\eta) = 
    \rho\prod_{w\in \eta} 
     g(x-w)
    \exp\left(
    \sum_{k=1}^\infty\frac{(-1)^k}{k!}
    \int_{(\R^d)^k}
    \left(
    \prod_{m=1}^k\ g(x-y_m)-1\right)
    \prod_{w\in\eta }\prod_{m=1}^k
     g(w-y_m) \rho_T^{(k)}(\yy_k)\dyy_k
    \right),
\end{align*}
for $x \in \R^d$ and $\eta \in\Gamma$. Note that, due to the conditions on $g$ the function $k$ is well-defined on $\R^d \times \Gamma$ and mild decay conditions on $g$. Since $k(x,\eta)= \exp (\mu-W(\{x\}\mid \eta))$, i.e.\ $k(x,\emptyset)= e^\mu$, it thus follows that $\mu$ is given by  
\begin{align*}
	\mu = \log \rho +
	\sum_{k=1}^\infty\frac{(-1)^k}{k!}
    \int_{(\R^d)^k}
    \left(
    \prod_{m=1}^k g(x-y_m)-1\right)
    \rho_T^{(k)}(\yy_k)\dyy_k.
\end{align*}
Furthermore, it is easily checked, that for the Kirkwood-closure process, there holds 
\begin{align}\label{eq:rhotildeeasy}
	\widetilde{\rho}_T^{\,(1+k)}(x,\yy_k) =
	\left(
    \prod_{m=1}^k g(x-y_m)-1\right)
    \rho_T^{(k)}(\yy_k),
\end{align}
thus assumption \ref{ass:A} is not needed in this case to identify the limit in \req{muexpansion}, although one could also show that it holds here.
\begin{corollary}
Let $g$ satisfy the assumptions above and $\P$ be the Kirkwood-closure process if 
\begin{align*}
	\rho< \frac{1}{(2+\zeta D)ebC(g)}
\end{align*}
then the chemical potential $\mu$ of the Kirkwood-closure process is given by \req{muexpansionfinal}.
\end{corollary}
\begin{remark}
In the case of the Kirkwood-closure it is quite easy to see the connection of the expansion \req{muexpansionfinal} to classic cluster expansions. It follows from \req{rhotildeeasy} and Section 3.1 \cite{Kuna07} that
\begin{align*}
	\widetilde{\rho}_T^{\,(1+k)}(\tilde{\yy}_{k+1}) 
	= \rho^{k} \sum_{C\in\mathcal{C}_{k+1}} 
	\prod_{\{i,j\}\in E(C)}
	\left(g(\tilde{y}_i-\tilde{y}_j)-1\right)
\end{align*}
where $\tilde{\yy}_{k+1}=(x,\yy_k)$. Thus, in this case we have
\begin{align*}
	\mu = \log \rho +
	\sum_{k=1}^\infty\frac{(-\rho)^{k}}{k!}
    \int_{(\R^d)^k}
	\sum_{C\in\mathcal{C}_{k+1}}
    \prod_{\{i,j\}\in E(C)}
     \left(g(\tilde{y}_i-\tilde{y}_j)-1\right)
     \dyy_k,
\end{align*}
where $\mathcal{C}_{k+1}$ is the set of connected graphs on $k+1$ vertices and $E(C)$ is the set of edges of the graph $C$. This has the same structure as the classic cluster expansion for the density in the case of pair-interactions with $-\rho$ taking the role of the activity (the exponential of the chemical potential) and $g-1$ taking the role of the Mayer function. 
\end{remark}
\begin{remark}
Using the easier structure of the family $(\widetilde{\rho}_T^{\,(1+k)})_{k\geq 1}$ in this case, i.e.\ \req{rhotildeeasy}, one can easily improve the bound \req{bellboundfinal} to
\begin{align*}
	\sup_{x\in\R^d}\int_{\Lambda^n}  \left|\widetilde{\rho}_T^{\,(1+k)}(x,\yy_k) \right|\dyy_k 
    \leq  k! M \left(\rho ebC(g)\right)^k
\end{align*}
and obtain a better radius of convergence.
\end{remark}

\section{A first step: The case of the zero-point density}\label{sec:firststep}
Let $\Lambda\subset \R^d$ be a bounded set, then in the case of $n=0$ we get from \req{JanossyHamiltonian} 
\begin{align*}
    j^{(0)}_\Lambda = \int_{\Gamma_{\Lambda^c}}\dP(\eta)    
\end{align*}
Now furthermore we know that by \req{Janossycorr}
\begin{align*}
      j_{ \Lambda}^{(0)}= 1+ \sum_{k=1}^\infty\frac{(-1)^k}{k!} \int_{\Lambda^k}\rho^{(k)}(\yy_k)\dyy_k.
\end{align*}
Our goal now is to write the above equation as an exponential, for that we use the {truncated correlation functions} defined in \req{clusterfunctions}. Note that in \req{clusterfunctions} we can easily multiply the left-hand side by $(-1)^n$ and split this factor into $\prod_{i=1}^k(-1)^{\kappa_i}$ on the right-hand side so that we get the representation \req{partitionrep} for $\Phi = (-1)^n\rho^{(n)}$ with $F_n= (-1)^n \rho_T^{(n)}$. \\
Assume that $\P$ satisfies assumption \ref{ass:B} for some $q<1/2$, then in particular
\begin{align}\label{eq:jzerobound}
    \int_\Lambda\int_{\Lambda^n}  \left|{\rho}_{T}^{(1+n)}(x,\yy_n) \right|\dyy_n\dx
    \leq |\Lambda|n! D\rho q^n.
\end{align}
This means the assumptions of Theorem \ref{thm:reb} are satisfied,
and we can write
\begin{align*}
    j_{ \Lambda}^{(0)}= \exp\left(\sum_{k=1}^\infty\frac{(-1)^k}{k!} \int_{\Lambda^k}\rho^{(k)}_T(\yy_k)\dyy_k\right)
\end{align*}
and thus 
\begin{align}\label{eq:j0exp}
    \log j_{ \Lambda}^{(0)}= 
    \sum_{k=1}^\infty\frac{(-1)^k}{k!} \int_{\Lambda^{k}}\rho^{(k)}_T(\yy_{k})\dyy_{k}.
\end{align}
\begin{remark}
Let us consider the same setting as in Example \ref{ex:superstab}. If in this case in \req{classiczbound} we have $\qb <1/3$, then $q= \qb/(1-\qb)<1/2$ and thus
taking the thermodynamic limit of \req{j0exp}, i.e.\ dividing by the volume $|\Lambda|$ and increasing $\Lambda$ to the whole space $\R^d$,  gives 
\begin{align*}
    \lim_{\Lambda\nearrow \R^d}\frac{1}{|\Lambda|}\log j_{ \Lambda}^{(0)}  = 
    -\rho+\sum_{k=1}^\infty\frac{(-1)^{k+1}}{(k+1)!} \int_{(\R^d)^{k}}\rho^{(k+1)}_T(0,\yy_{k})\dyy_{k}=-p(\mu,u),
\end{align*}
where $p$ is the infinite-volume grand-canonical pressure, cf.\ \cite{Penrose63}. The above expression follows immediately from the cluster expansions of the pressure and the truncated correlation functions therein and \cite{Ruelle64}.
\end{remark}

\section{Generalization to the one-point density}\label{sec:proof}
We now want to generalize this approach for the one-particle density $j^{(1)}_\Lambda$ of an infinite-volume Gibbs measure. Again we will use \req{Janossycorr} to find an expansion and \req{JanossyHamiltonian} to identify the terms. By \req{Janossycorr} we have
\begin{align}\label{eq:j1exp}
    j^{(1)}_\Lambda(x)= \sum_{k=0}^\infty\frac{(-1)^k}{k!} \int_{\Lambda^k}\rho^{(1+k)}(x,\yy_k)\dyy_k =  \rho\left(1+\sum_{k=1}^\infty\frac{(-1)^k}{k!} \int_{\Lambda^k}\frac{\rho^{(k+1)}(x,\yy_k)}{\rho}\dyy_k\right).
\end{align}
We introduce the family $(F_k)_{k\geq1}$ that is recursively defined by $F_1(x,y)  = \rho^{(2)}(x,y)/\rho$ and for $k\geq 2$ by 
\begin{align}\label{eq:fnk}
    F_k(x,\yy_k) =
    \frac{\rho^{(1+k)}(x,\yy_{k })}{\rho}
    -\sum_{l=2}^{k }\sum_{\pi\in\Pi_l(\{\yy_k\})}\prod_{i=1}^l F_{ \kappa_i}(x,\pi_i).
\end{align}
{\color{black} The ansatz \req{j1exp} with \req{fnk} was first used by Nettleton and Green in \cite{Nettleton57} 
and the division by $\rho$ lets us use Theorem \ref{thm:reb} to obtain the exponential representation in \req{j1exp} as the series then starts at one.} As previously noted, we have $j^{(1)}_\Lambda \to 0 $ as $\Lambda\nearrow \R^d$, thus we need to identify the divergent part of $\log j^{(1)}_\Lambda$. 
\begin{proposition}
    For every $k \in\mathbb{N}$ there holds
\begin{align}\label{eq:goneksplit}
    F_k(x,\yy_k)=  \rho_T^{(k)}(\yy_k)+ \widetilde{\rho}^{\,(1+k)}_T(x,\yy_k),
\end{align}
where the family $(\widetilde{\rho}^{\,(1+k)}_T)_{k\geq 1}$ is the one defined in \req{rhotilde}.
\end{proposition}


\begin{proof}
The claim holds for $k=1$. By definition \req{fnk} of the $F_k$ there holds
\begin{align}\label{eq:gonek}
F_{k+1}(x,\yy_{k+1}) = \frac{\rho^{(k+2)}(x,\yy_{k+1})}{\rho} - \sum_{l=2}^{k+1}\sum_{\pi\in\Pi_l(\{\yy_{k+1}\})}\prod_{i=1}^l F_{ \kappa_i}(x,\pi_i).
\end{align}
For the second part of the right-hand side above we can use the induction hypothesis to find
\begin{align*}
    \sum_{l=2}^{k+1}\sum_{\pi\in\Pi_l(\{\yy_{k+1}\})}\prod_{i=1}^l F_{ \kappa_i}(x,\pi_i)
    =\sum_{l=2}^{k+1}\sum_{\pi\in\Pi_l(\{\yy_{k+1}\})}\prod_{i=1}^l (\rho_T^{(\kappa_i)}(\pi_i)+ \widetilde{\rho}^{\,(1+\kappa_i)}_T(x,\pi_i))
\end{align*}
and thus we get
\begin{align}\label{eq:fpluggedin}
    F_{k+1}(x,\yy_{k+1}) = \frac{\rho^{(k+2)}(x,\yy_{k+1})}{\rho}-
    \sum_{l=2}^{k+1}\sum_{\pi\in\Pi_l(\{\yy_{k+1}\})}\prod_{i=1}^l (\rho_T^{(\kappa_i)}(\pi_i)+ \widetilde{\rho}^{\,(1+\kappa_i)}_T(x,\pi_i)).
\end{align}
For the first term of the right-hand side of \req{gonek}, there holds by \req{clusterfunctions}
\begin{align*}
    \frac{\rho^{(k+2)}(x,\yy_{k+1})}{\rho} = \frac{1}{\rho} \sum_{l=1}^{k+2}\sum_{\pi\in\Pi_l(\{x,\yy_{k+1}\})}\prod_{i=1}^l
    \rho_T^{(\kappa_i)}(\pi_i).
\end{align*}
For any partition $\pi\in\Pi_l(\{x,\yy_{k+1}\})$  we denote by $\nu$ the index such $\pi_\nu=(x,\pi_\nu')$, i.e.\ $\pi_\nu$ is the part of the partition $\pi$ containing $x$. We can thus write
\begin{align}\label{eq:factoroutx}
    \frac{\rho^{(k+2)}(x,\yy_{k+1})}{\rho} = 
   \sum_{l=1}^{k+2}\sum_{\pi\in\Pi_l(\{x,\yy_{k+1}\})} \frac{\rho^{(\kappa_\nu)}_T( x,{\pi_\nu'})}{\rho}\prod_{i=1 \atop i \neq \nu}^l
    \rho^{(\kappa_i)}_T( {\pi_i}).
\end{align}
We distinguish three cases:
\begin{enumerate}
    \item If $\kappa_{\nu}=k+2$, then $\pi $ is the partition into one element and contributes the term $\frac{\rho^{(k+2)}_T(x,\yy_{k+1})}{\rho}$ which only appears in the case $l=1$ above;
    \item If $\kappa_{\nu}=1$, then ${\rho^{(\kappa_\nu)}_T}/{\rho}=1$ and we can understand $\widehat{\pi}:=\pi\backslash\pi_\nu$ as a partition of the elements $\yy_{k+1}$ into $l-1$ elements, i.e.\ an element of $\widehat{\pi}\in \Pi_{l-1}(\yy_{k+1})$. Summing over all of the partitions where $\kappa_{\nu}=1$ then gives
    \begin{align*}
        \sum_{l=2}^{k+2}
        \sum_{\pi\in\Pi_l(\{x,\yy_{k+1}\})\atop \kappa_\nu=1}
        \prod_{i=1 \atop i\neq \nu}^l
        \rho_T^{(\kappa_i)}(\pi_i)
        =\sum_{l=2}^{k+2}
        \sum_{\widehat{\pi}\in\Pi_{l-1}( \yy_{k+1} )}
        \prod_{i=1}^{l-1}
        \rho_T^{(\widehat{\kappa}_i)}(\widehat{\pi}_i)
    \end{align*}
    where $\widehat{\kappa}_i$ is the number of elements in $\widehat{\pi}_i$.
    Shifting the index we get
    \begin{align*}
        \sum_{l=2}^{k+2}
        \sum_{\widehat{\pi}\in\Pi_{l-1}( \{\yy_{k+1} \})}
        \prod_{i=1}^{l-1}
        \rho_T^{(\widehat{\kappa}_i)}(\widehat{\pi}_i)=
        \sum_{l=1}^{k+1}
        \sum_{\widehat{\pi}\in\Pi_{l}( \yy_{k+1} )}
        \prod_{i=1}^{l}
        \rho_T^{(\widehat{\kappa}_i)}(\widehat{\pi}_i)= \rho^{(k+1)}_T(\yy_{k+1})
    \end{align*}
    where we have used \req{clusterfunctions} for the last equality.
    \item If $ \kappa_{\nu}=m+1$ for some $1\leq m\leq k$, we can use the induction assumption to write
    \begin{align*}
        \frac{\rho^{(1+m)}_T(x,\pi_\nu')}{\rho} = 
    \sum_{j=1}^{m} \sum_{\widehat{\pi}\in\Pi_l(\{\pi_\nu'\})} \prod_{s=1}^j
    \widetilde{\rho}^{\,(1+\widehat{\kappa}_s)}_T (x,{\widehat{\pi}_s}).
    \end{align*}
\end{enumerate}
Plugging the three cases above into \req{factoroutx} we find
\begin{align*}
    &\frac{\rho^{(k+2)}(x,\yy_{k+1})}{\rho} = 
   \sum_{l=1}^{k+2}\sum_{\pi\in\Pi_l(\{x,\yy_{k+1}\})}\sum_{\kappa_\nu=1}^{k+2} \frac{\rho^{(\kappa_\nu)}_T( x,{\pi_\nu'})}{\rho}\prod_{i=1 \atop i \neq \nu}^l
    \rho^{(\kappa_i)}_T( {\pi_i}) 
    \\
    &=\frac{\rho^{(k+2)}_T(x,\yy_{k+1})}{\rho}+\rho^{(k+1)}_T(\yy_{k+1})+\sum_{l=2}^{k+2}\sum_{\pi\in\Pi_l(\{x,\yy_{k+1}\})}\sum_{\kappa_\nu=2}^{k+1} \sum_{j=1}^{m} \sum_{\widehat{\pi}\in\Pi_l(\{\pi_\nu'\})} \prod_{s=1}^j
    \widetilde{\rho}^{\,(1+\widehat{\kappa}_s)}_T (x,{\widehat{\pi}_s})\prod_{i=1 \atop i \neq \nu}^l
    \rho^{(\kappa_i)}_T( {\pi_i}). 
\end{align*}
Finally, we plug the above into \req{fpluggedin} to get
\begin{align*}
    F_{k+1}(x,\yy_{k+1}) &= \frac{\rho^{(k+2)}_T(x,\yy_{k+1})}{\rho}+\rho^{(k+1)}_T(\yy_{k+1})
    \\
    &+\sum_{l=2}^{k+2}\sum_{\pi\in\Pi_l(\{x,\yy_{k+1}\})}\sum_{\kappa_\nu=2}^{k+1} \sum_{j=1}^{m} \sum_{\widehat{\pi}\in\Pi_l(\{\pi_\nu'\})} \prod_{s=1}^j
    \widetilde{\rho}^{\,(1+\widehat{\kappa}_s)}_T (x,{\widehat{\pi}_s})\prod_{i=1 \atop i \neq \nu}^l
    \rho^{(\kappa_i)}_T( {\pi_i})
    \\
    &-
    \sum_{l=2}^{k+1}\sum_{\pi\in\Pi_l(\{\yy_{k+1}\})}\prod_{i=1}^l (\rho_T^{(\kappa_i)}(\pi_i)+ \widetilde{\rho}^{\,(1+\kappa_i)}_T(x,\pi_i)).
\end{align*}
Looking closer at the difference of the last two lines above, we see that all terms containing at least one $\rho_T^{(\kappa_i)}$ and one $\widetilde{\rho}^{\,(1+\kappa_j)}_T$ cancel and we get
\begin{align*}
    &\sum_{l=2}^{k+2}\sum_{\pi\in\Pi_l(\{x,\yy_{k+1}\})}\sum_{\kappa_\nu=2}^{k+1} \sum_{j=1}^{m} \sum_{\widehat{\pi}\in\Pi_l(\{\pi_\nu'\})} \prod_{s=1}^j
    \widetilde{\rho}^{\,(1+\widehat{\kappa}_s)}_T (x,{\widehat{\pi}_s})\prod_{i=1 \atop i \neq \nu}^l
    \rho^{(\kappa_i)}_T( {\pi_i})
    \\
    &-
    \sum_{l=2}^{k+1}\sum_{\pi\in\Pi_l(\{\yy_{k+1}\})}\prod_{i=1}^l (\rho_T^{(\kappa_i)}(\pi_i)+ \widetilde{\rho}^{\,(1+\kappa_i)}_T(x,\pi_i))
    \\
    &=-\sum_{l=2}^{k+1}\sum_{\pi\in\Pi_l(\{\yy_{k+1}\})}\prod_{i=1}^l \widetilde{\rho}^{\,(1+\kappa_i)}_T(x,\pi_i).
\end{align*}
We thus have that 
\begin{align*}
    F_{k+1}(x,\yy_{k+1}) &= \rho^{(k+1)}_T(\yy_{k+1})
    +\frac{\rho^{(k+2)}_T(x,\yy_{k+1})}{\rho}
    -\sum_{l=2}^{k+1}\sum_{\pi\in\Pi_l(\{\yy_{k+1}\})}\prod_{i=1}^l \widetilde{\rho}^{\,(1+\kappa_i)}_T(x,\pi_i).
\end{align*}
Finally, using \req{rhotilde} proves \req{goneksplit}.

\end{proof}
We now want to apply Theorem \ref{thm:reb} to \req{j1exp}, in view of \req{goneksplit} and \req{Rboundcluster} we thus need a bound on the integrals over $(\widetilde{\rho}_T^{\,(1+k)})_{k\geq 1}$.
\begin{theorem}\label{thm:bellbound}
Let $\P$ be a $(\mu,H)$-Gibbs measure that satisfies Assumption \ref{ass:B}, then for any bounded $\Lambda \subset \R^d$
\begin{align}\label{eq:bellboundfinal}
    \sup_{x\in\R^d}\int_{\Lambda^n}  \left|\widetilde{\rho}_T^{\,(1+k)}(x,\yy_k) \right|\dyy_k 
    \leq  (k-1)! M (1+\zeta D)^k q^k
\end{align}
for some $M>0$ and $\zeta>1$ independent of $\Lambda$.
\end{theorem}
\begin{proof}
By \req{rhotilde} there holds
\begin{align*}
    \int_{\Lambda^k} \widetilde{\rho}_T^{\,(1+k)}(x,\yy_k)  \dyy_k
    =
    \int_{\Lambda^k} \frac{\rho^{(1+k)}_T(x,\yy_k)}{\rho} \dyy_k 
    -\sum_{l=2}^k \sum_{\pi\in\Pi_l(\{\yy_k\})} \prod_{i=1}^l\int_{\Lambda^{\kappa_i}} \widetilde{\rho}^{\,(1+\kappa_i)}_T (x,\pi_i) \diff {\pi_i}.
\end{align*}
Note that because we integrate over all $\yy_k$, when looking at a particular partition $\pi\in\Pi_l(\{\yy_k\})$ it does not matter which particular elements of $\yy_k$ are contained in which part $\pi_i$ of $\pi$ only the sizes of the different parts $\pi_i$ matter. This is the definition of the $k$th Bell polynomial, cf.\ \cite{Comtet74}, and allows us to write
\begin{align*}
    \int_{\Lambda^k} \widetilde{\rho}_T^{\,(1+k)}(x,\yy_k)  \dyy_k
    =
    -B_k\left(
    I_1,\dots,I_{k-1},-\int_{\Lambda^k} \frac{\rho^{(1+k)}_T(x,\yy_k)}{\rho} \dyy_k
    \right)
\end{align*}
where $I_i=\int_{\Lambda^{\kappa_i}} \widetilde{\rho}^{\,(1+\kappa_i)}_T (x,\yy_{\kappa_i}) \dyy_{\kappa_i}$. Taking the absolute value above and using the bound from Assumption \ref{ass:B} we thus get
\begin{align}\label{eq:bellbound}
    \int_{\Lambda^k}  \left|\widetilde{\rho}_T^{\,(1+k)}(x,\yy_k) \right|\dyy_k 
    \leq w_k 
\end{align}
where the $w_k$ are recursively defined by $w_1=Dq$ and for $k\geq 2$
\begin{align}\label{eq:wk}
    w_k=B_k \left( {w}_{1},\dots, {w}_{k-1},k!Dq^k\right).
\end{align}
We use the ansatz
\begin{align*}
    w_k = q^k\sum_{l=0}^k a_l^{(k)}D^l,
\end{align*}
and want to calculate the coefficients $a_l^{(k)}$ for $0\leq l\leq k$ in $w_k$ using the well-known recursion, cf.\ \cite{Comtet74},
\begin{align}\label{eq:bellrecursion}
    B_{k+1}\left(x_1 ,\dots,x_{k+1}\right) =
    \sum_{i=0}^{k} \binom{k}{i}B_{k-i}(x_1,\dots,x_{k-i})x_{i+1},
\end{align}
where $B_0:=1$. This recursion along with the fact that $w_1=Dq$ already implies that $a_0^{(k)}=0$ for every $k$. 
Lastly, we note that \req{bellrecursion} implies that
\begin{align*}
    B_k(x_1,\dots,x_{k-1},x_k+y_k) =B_k(x_1,\dots,x_{k-1},x_k)+y_k.
\end{align*}
Putting $x_i=w_i$ $1\leq i \leq k-1$, $x_k=q^{k}k!D$ and $y_k=w_k-q^{k}k!D$ then gives
\begin{align*} 
    B_{k}(w_1,\dots,w_{k}) = q^{k}\left(2\sum_{l=1}^{k} a_l^{(k)}D^l-k!D\right).
\end{align*}
Plugging the above into \req{bellrecursion} we find
\begin{align*}
    B_{k+1}\left(w_1 ,\dots,w_{k},(k+1)!Dq^{k+1}\right)
    &=(k+1)!Dq^{k+1} \\
    &+\sum_{i=0}^{k-1} \binom{k}{i}
    q^{k-i}\left(2\sum_{l=1}^{k-i} a_l^{(k-i)}D^l-(k-i)!D\right)
    q^{i+1}\left(\sum_{l=1}^{i+1} a_l^{(i+1)}D^l\right).
\end{align*}
This already shows that $a^{(k+1)}_1= (k+1)!$ as in all the terms resulting from the second sum $D$ has an exponent of at least two. We can thus use the above equality to find
\begin{align*}
    &a_2^{(k+1)}=\sum_{i=0}^{k-1} \binom{k}{i}(k-i)!(i+1)!, \\
    &a_m^{(k+1)}=
    \sum_{i=0}^{k-1} \binom{k}{i}\left(
    (k-i)! a_{m-1}^{(i+1)}+2(i+1)! a_{m-1}^{(k-i)}
    +2\sum_{\nu=2}^{m-2}  a_\nu^{(i+1)}a_{m-\nu}^{(k-i)}
    \right), \qquad m \geq 3.
\end{align*}
Straightforward computation gives
\begin{align*}
    a_2^{(k+1)}= \sum_{i=0}^{k-1}\binom{k}{i}(k-i)!(i+1)! = (k+1)!\frac{k}{2}. 
\end{align*}
We claim that for every $k,\nu\in\mathbb{N}$ there holds
\begin{align}\label{eq:coeffseq}
    a_\nu^{(k)}  = \frac{k!}{\nu!(\nu-1)!} (k-1)\cdot\ldots \cdot(k-\nu+1)b_\nu
    = \frac{k!}{\nu!}\binom{k-1}{\nu-1}b_\nu
\end{align}
where $(b_\nu)_{\nu\geq 0}$ is the sequence defined by $b_0=0, b_1=b_2=1$ and for $\nu\geq 2$
\begin{align}\label{eq:brecursion}
    b_{\nu+1}= (\nu+2)b_\nu + 2 \sum_{j=2}^{\nu-1} \binom{\nu}{j}b_j b_{\nu-j+1}.
\end{align}
Note that $a_m^{(k)}=0$ if $k<m$.
We have already established that \req{coeffseq} holds for all $k\geq m$ for $m=0,1,2$. 
We will prove the claim by induction. Let \req{coeffseq} hold for all $k$ and every $\nu\leq m$ for some $m\geq 2$, then 
\begin{align}\nonumber
    a_{m+1}^{(k+1)}
    & =\sum_{i=0}^{k-1}\binom{k}{i}(k-i)!\,a_{m}^{(i+1)}
    + 2\sum_{i=0}^{k-1} \binom{k}{i} 
    \sum_{\nu=1}^{m-1} a_\nu^{(i+1)}a_{m+1-\nu}^{(k-i)} \\
    & =\sum_{i=0}^{k-1}\binom{k}{i}(k-i)!\frac{(i+1)!}{m!}\binom{i}{m-1} b_m \label{eq:firstsummand} \\
    &+ 2\sum_{i=0}^{k-1} \binom{k}{i} 
    \sum_{\nu=1}^{m-1} \frac{(i+1)!}{\nu!}\binom{i}{\nu-1} b_\nu
    \frac{(k-i)!}{(m+1-\nu)!}\binom{k-i-1}{m-\nu} b_{m+1-\nu} .
    \label{eq:secondsummand}
\end{align}
We start by simplifying the term in \req{firstsummand} and get
\begin{align}\label{eq:firstsumsimple}
    \sum_{i=0}^{k-1}\binom{k}{i}(k-i)!\frac{(i+1)!}{m!}\binom{i}{m-1} b_m
    =\frac{k!b_m}{m!}m\sum_{i=0}^{k-1}\binom{i+1}{m}
    =\frac{k!b_m}{m!} \binom{k+1}{m+1}m
\end{align}
by the Vandermonde identity, cf.\ \cite{spiess90}. 
For the term in \req{secondsummand}, we first swap the inner and outer sums and rearrange the binomial coefficients to get
\begin{align*}
    &\sum_{i=0}^{k-1} \binom{k}{i} 
    \sum_{\nu=1}^{m-1}  \frac{(i+1)!}{\nu!}\binom{i}{\nu-1} b_\nu
     \frac{(k-i)!}{(m+1-\nu)!}\binom{k-i-1}{m-\nu} b_{m+1-\nu} \\
    &=\frac{k!}{m!}\sum_{\nu=1}^{m-1}
    \binom{m}{\nu-1}
    b_\nu b_{m+1-\nu}\sum_{i=0}^{k-1}  
      \binom{i+1}{\nu}\binom{k-i-1}{m-\nu}.
\end{align*}
We note that 
\begin{align*}
    \sum_{i=0}^{k-1}  
      \binom{i+1}{\nu}\binom{k-i-1}{m-\nu}
     =\binom{k+1}{m+1},
\end{align*}
which is a version of the Chu-Vandermonde identity, see e.g.
\cite{spiess90}, and thus
\begin{align}\label{eq:secondsumsimple}
    \frac{k!}{m!}\sum_{\nu=1}^{m-1}
    \binom{m}{\nu-1}
    b_\nu b_{m+1-\nu}\sum_{i=0}^{k-1}  
      \binom{i+1}{\nu}\binom{k-i-1}{m-\nu}
      =\frac{k!}{m!}\binom{k+1}{m+1} \sum_{\nu=1}^{m-1}
    \binom{m}{\nu-1}
    b_\nu b_{m+1-\nu}. 
\end{align}
Plugging \req{firstsumsimple} and \req{secondsumsimple} back into \req{firstsummand} and \req{secondsummand} we get 
\begin{align*}
    a_{m+1}^{(k+1)}= \frac{k!}{m!} \binom{k+1}{m+1}
    \left(m b_m+2 b_m +2\sum_{\nu=2}^{m-1}\binom{m}{\nu}b_\nu b_{m+1-\nu}\right) =  \frac{k!}{m!} \binom{k+1}{m+1} b_{m+1}
\end{align*}
where we used that 
\begin{align*}
    \sum_{\nu=1}^{m-1}\binom{m}{\nu-1}b_\nu b_{m+1-\nu} 
    =b_m+\sum_{\nu=2}^{m-1}\binom{m}{\nu}b_\nu b_{m+1-\nu}.
\end{align*}
The claim is thus proved, and we arrive at the bound
\begin{align*}
    \int_{\Lambda^k}  \left|\widetilde{\rho}_T^{\,(1+k)}(x,\yy_k) \right|\dyy_k 
    \leq
    q^k\sum_{m=1}^k\binom{k}{m} \frac{(k-1)!}{(m-1)!} b_{m}D^m.
\end{align*}
The sequence $(b_m)_{m\geq 0}$ is the number of total partitions of $m$ elements, see \cite{Flajolet97} and \cite{Flajolet09}. Therein we also find the asymptotic behavior of $(b_m)_m$, namely, there is an $M>0$ such that
\begin{align*}
    b_{m} \leq M\,\frac{m^{m-1}\zeta^m}{e^m} \qquad \text{ where }\qquad \zeta= \frac{1}{2\log 2-1},
\end{align*}
giving
\begin{align*}
    \sum_{m=1}^k\binom{k}{m} \frac{(k-1)!}{(m-1)!} b_{m}D^m \leq
    M (k-1)! (1+\zeta D)^k.
\end{align*}

\end{proof}
\begin{proofthm}
Assume that $\P$ is a $(\mu,H)$-Gibbs measure satisfying a Ruelle-condition as well as Assumptions \ref{ass:A} and \ref{ass:B}. By \req{goneksplit}, the triangle inequality, \req{jzerobound} and Theorem \ref{thm:bellbound} we get for $F_n$ of \req{fnk}, that
\begin{align*}
	\int_{\Lambda^n}\left|F_n(0,\yy_n)\right| \dyy_n 
	\leq |\Lambda|(n-1)! D \rho q^{n-1}+M (n-1)! (1+\zeta D)^n q^n
	\leq \left( |\Lambda|\frac{D \rho }{q}+M\right)(n-1)! \left(q(2+\zeta D)\right)^n.
\end{align*} 
If now $q<q_0$ from \req{qzero}, then by Theorem \ref{thm:reb} there holds
\begin{align*}
   \log j_\Lambda^{(1)}(0) = \log \rho  +  
   \sum_{k=1}^\infty\frac{(-1)^k}{k!} \int_{\Lambda^k}
   \left(\rho_T^{(k)}(\yy_k)+\widetilde{\rho}^{\,(1+k)}_T(0,\yy_k)\right)\dyy_k.
\end{align*}
Which, in view of \req{j0exp}, is equivalent to
\begin{align}\label{eq:muexpansion}
    \log j_\Lambda^{(1)}(0)-\log j_\Lambda^{(0)}
    =\log \rho +  
   \sum_{k=1}^\infty\frac{(-1)^k}{k!} \int_{\Lambda^k}
   \widetilde{\rho}^{\,(1+k)}_T(0,\yy_k) \dyy_k.
\end{align}
By Assumption \ref{ass:A} and \req{JanossyHamiltonian} the left-hand side of \req{muexpansion} converges to $\mu$ as $\Lambda \nearrow\R^d$ and since the bound in \req{bellboundfinal} is independent of $\Lambda$ the right-hand side above converges by dominated convergence, and we arrive at \req{muexpansionfinal}.  
\end{proofthm}

\section*{Declarations}
Data sharing is not applicable to this article as no datasets were generated or 
analyzed during the current study. The author states that there is no 
conflict of interest.


\end{document}